\documentclass[preprint,18pt]{elsarticle}
\usepackage{natbib}
\usepackage{lineno}
\usepackage{url,graphicx}
\usepackage{amsmath,amsfonts,amssymb}
\usepackage{amsthm}
\usepackage{graphicx}
\usepackage{caption}
\usepackage{subcaption}
\usepackage{tikz}
\usetikzlibrary{arrows,petri,topaths}
\usepackage{tkz-berge}

\usepackage{todonotes}

\usepackage{hyperref}
\hypersetup{
  pdfauthor={Zijia Li, Josef Schicho, Hans-Peter Schr\"ocker},%
  pdftitle={Factorization of Motion Polynomials},%
  colorlinks=true,%
}

\usepackage{algorithm}
\usepackage{algpseudocode}

\algrenewcommand\algorithmicwhile{\textbf{While}}
\algrenewcommand\algorithmicfor{\textbf{For}}
\algrenewcommand\algorithmicdo{\textbf{Do}}
\algrenewcommand\algorithmicif{\textbf{If}}
\algrenewcommand\algorithmicthen{\textbf{Then}}
\algrenewcommand\algorithmicelse{\textbf{Else}}
\algrenewcommand\algorithmicend{\textbf{End}}
\algrenewcommand\algorithmicreturn{\textbf{Return}}

\newtheorem{thm}{Theorem}
\newtheorem{lem}{Lemma}

\theoremstyle{definition}
\newtheorem{mydef}{Definition}
\theoremstyle{remark}
\newtheorem{example}{Example}

\newcommand{\R}{\mathbb{R}}

\def\H{\mathbb H}
\def\cj{\overline}
\newcommand{\SE}[1][3]{\mathrm{SE}(#1)}
\newcommand{\D}{\mathbb{D}}

\newcommand{\qi}{\mathbf{i}}
\newcommand{\qj}{\mathbf{j}}
\newcommand{\qk}{\mathbf{k}}
\newcommand{\eps}{\epsilon}
\newcommand{\al}{\alpha}
\newcommand{\be}{\beta}
\newcommand{\ga}{\gamma}
\newcommand{\itstep}[1]{\par\bigskip\noindent\textbf{#1:}}

\DeclareMathOperator{\GRPF}{GRPF}
\DeclareMathOperator{\comp}{comp}

\begin{document}

\begin{frontmatter}

\title{Factorization of Motion Polynomials}

\author[zijiaadd]{Zijia Li\corref{mycorrespondingauthor}}

\cortext[mycorrespondingauthor]{Corresponding author}
\ead{zijia.li@oeaw.ac.at}

\author[josefadd]{Josef Schicho}

\author[hanspeteradd]{Hans-Peter Schr\"ocker}

\address[zijiaadd]{Johann Radon Institute for Computational and Applied Mathematics (RICAM), 
Austrian Academy of Sciences, 
Altenberger Strasse 69, 
4040 Linz, Austria}
\address[josefadd]{Research Institute for Symbolic Computation, 
Johannes Kepler University, 
Altenberger Strasse 69, 
A-4040 Linz, Austria}
\address[hanspeteradd]{Unit Geometry and CAD, Faculty of Engineering Science,
University Innsbruck, Technikerstrasse 13, 6020 Innsbruck, Austria}

\begin{abstract}
  In this paper, we consider the existence of a factorization of a
  monic, bounded motion polynomial.  We prove existence of
  factorizations, possibly after multiplication with a real polynomial
  and provide algorithms for computing polynomial factor and
  factorizations. The first algorithm is conceptually simpler but may
  require a high degree of the polynomial factor.  The second
  algorithm gives an optimal degree. 
\end{abstract}

\begin{keyword}
Dual quaternion; Study Quadric; Rational Motion; Linkage
\end{keyword}

\end{frontmatter}

\section{Introduction}
\label{intr}

Let $\H[t]$ be the ring of univariate polynomials with quaternion coefficients,
with the variable $t$ commuting with the coefficients.
The existence of factorizations of quaternion polynomials into linear factors is a classical 
 result \cite{GM}. 
In \cite{part11}, motion polynomials are defined as elements of $\D\H[t]$ -- the ring of univariate 
polynomials with dual quaternions -- with real norms; these can be
used to parametrize rational motions in Euclidean 3-space. The main result there is that a factorization
into linear factors allows to construct a mechanical linkage that generates the desired motion. An adaption
of the algorithm by \cite{GM} to the dual quaternion case indeed allows to factorize ``generic'' polynomials
in $\D\H[t]$, namely those whose primal part has no strictly real factors. For fixed degree, the set
of generic motion polynomials is open and dense in the set of all motion polynomials. 

Since 2012, we have been wondering which non-generic motion polynomials do allow factorization into linear
factors. One reason for our curiousity is a paradoxical fact: rational motions that are parametrized by
generic motion polynomials have special properties, namely that their orbit curves have full cyclicity.
The question is still not completely solved, but in this paper we give an affirmative answer for ``bounded''
motion polynomials. They always admit factorizations into products of linear rotation polynomials, 
 possibly after multiplication with a real polynomial. 
 This changes the motion polynomial but not the motion it parameterizes. Bounded motion polynomials are 
 defined by the condition that the norm polynomials has no real roots. The
kinematic meaning of this condition is that the orbits are bounded curves. It is also quite obvious that
motions that can be generated by linkages with revolute joints (in particular, no translational joints) are bounded,
hence the results in this paper cover all cases for which there is a linkage with revolute joints.

The results in this paper have been influenced by our study \cite{SL} of linkages producing straight line
motions and by the factorization of planar bounded motion polynomials given in \cite{planarf}.

The paper presents two different factorization algorithms for motion polynomials (Algorithms~\ref{talg} and \ref{dtalg}). 
 Both compute a co-factor $Q \in \R[t]$ and a factorization of $QM$
 for a given bounded motion polynomial $M \in \D\H[t]$. 
 The difference between these two algorithms is that Algorithm~\ref{talg} is conceptually simpler, 
 but it is calling Algorithm~\ref{p0alg}, introduced by \cite{planarf} for the factorization of 
 planar motion polynomials. Algorithm~\ref{dtalg} is more complicated but it produces co-factors 
 of optimal degree and does not depend on Algorithm~\ref{p0alg}. It is capable of producing 
 non-planar factorizations of planar motion polynomials but can be specialized to yield planar 
 factorizations as well.

\paragraph*{Structure of the paper} The remaining part of the paper is set up as follows. In Section 2,
 we recall the basic notations of dual quaternions and motion polynomials. 
 Section 3 will focus on the factorization of a motion polynomial. We also give some concrete 
 examples to support the algorithm.

\section{Dual quaternions and motion polynomials}

We start with a brief introduction to the dual quaternion model of rigid body displacements. 
 In particular, we focus on one degree of freedom rational motions that can be parameterized 
 by motion polynomials~\cite{PT}.

The dual quaternions form an eight-dimensional associative algebra over the real numbers. This algebra is generated by the base elements
\begin{equation*}
  1,\quad
  \qi,\quad
  \qj,\quad
  \qk,\quad
  \eps,\quad
  \eps\qi,\quad
  \eps\qj,\quad
  \eps\qk
\end{equation*}
and multiplication is defined via the relations
\begin{equation*}
  \qi^2 = \qj^2 = \qk^2 = \qi\qj\qk = -1,\quad
  \eps^2 = 0,\quad
  \qi\eps = \eps\qi,\quad
  \qj\eps = \eps\qj,\quad
  \qk\eps = \eps\qk.
\end{equation*}
The set of dual quaternions is denoted by $\D\H$, the sub-algebra of quaternions $\H$ is generated by $1$, $\qi$, $\qj$, and $\qk$.  Any dual quaternion can be written as $h = p + \eps q$ with $p, q \in \H$.  The \emph{conjugate dual quaternion} is $\cj{h} = \cj{p} + \eps \cj{q}$ and conjugation of quaternions is done by multiplying the coefficients of $\qi$, $\qj$, and $\qk$ with $-1$.

The \emph{norm} of the dual quaternion $h$ is defined as $\Vert h
\Vert = h\cj{h}$. It equals $p\cj{p} + \eps(p\cj{q}+q\cj{p})$ and is a
dual number (an element of the sub-algebra $\D$ generated by $1$ and $\eps$). A dual quaternion of norm $1$ is called a \emph{unit dual quaternion.}

Dual quaternions have important applications in kinematics and mechanism science. This is due to an isomorphism between the factor group of unit dual quaternions modulo $\pm 1$ and $\SE[3]$, the group of rigid body displacements. The rigid body displacement described by $h = p + \eps q$ with $h\cj{h} = 1$ maps the point $x = x_1\qi + x_2\qj + x_3\qk$ to
\begin{equation*}
  px\cj{p} + p\cj{q} - q\cj{p} = px\cj{p} + 2p\cj{q}.
\end{equation*}

Denote by $\D\H[t]$ the ring of polynomials in $t$ with dual quaternion coefficients where multiplication is defined by the convention that the indeterminate $t$ commutes with all coefficients. We follow the convention to write the coefficients to the left of the indeterminate $t$. Similarly, we denote by $\H[t]$ the sub-ring of polynomials with coefficients in~$\H$. The \emph{conjugate polynomial} to $C = \sum_{i=0}^n c_it_i \in \D\H[t]$ is $\cj{C} = \sum_{i=0}^n \cj{c}_it^i$ and the \emph{norm polynomial} is $C\cj{C}$. Its coefficients are dual numbers.  If $C = \sum_{i=0}^n c_it^i$, the value $C(h)$ of $C$ at $h \in \D\H$ is defined as $C(h) = \sum_{i=0}^n c_ih^i$. We also define $C(\infty) := c_n$.

A polynomial $M = P + \eps Q \in \D\H[t]$ is called a \emph{motion
  polynomial} if $P\cj{Q} + Q\cj{P} = 0$ and its leading coefficient
is invertible. Usually we will even assume that the leading coefficient
is one (the polynomial is \emph{monic}). This can be accomplished by left-multiplying $M$ with the inverse of the leading coefficient and often constitutes no loss of generality. The defining conditions of a motion polynomial ensure that its norm polynomial has real coefficients.

A motion polynomial $M = P + \eps Q$ acts on a point $x = x_1\qi + x_2\qj + x_3\qk$ according to
\begin{equation}
  \label{eq:1}
  x \mapsto \frac{Px\cj{P} + 2P\cj{Q}}{P\cj{P}}.
\end{equation}
This equation defines a rigid body displacement for all values
$t \in \R \cup \{\infty\}$ that are not zeros of $P$. Any map of the
shape \eqref{eq:1} with a motion polynomial $M = P + \eps Q$ is called
a \emph{rational motion.} We also say that the motion polynomial
\emph{parameterizes} the rational motion. The motion's trajectories
(orbits of points for $t \in \R \cup \{\infty\}$) are rational
curves. It is known that any motion with only rational trajectories is
parameterized by a motion polynomial
\cite{juettler93:_rationale_bewegungsvorgaenge}.

The simplest motion polynomials are of degree one and can be written as $M = t - h$ where $h - \cj{h} \in \R$ and $h\cj{h} \in \R$. They parameterize either rotations about a fixed axis or translations in a fixed direction. We speak of the \emph{rotation} or \emph{translation quaternion} $h$ and the \emph{rotation} or \emph{translation polynomial} $t-h$, respectively. In this paper we are concerned with the factorization of motion polynomials into the product of rotation polynomials. These are distinguished from translation polynomials by having a primal part not in~$\R[t]$.

\section{Factorization}

In \cite{part11} it has been shown that a generic monic motion polynomial $M = P + \eps D$ of degree $n$ admits factorizations of the shape
\begin{equation}
  \label{eq:2}
  M = (t - h_1) \cdots (t - h_n)
\end{equation}
with rotation polynomials $t-h_1, \ldots, t - h_n$. Here, the term ``generic'' means that the primal part $P$ of $M$ has no real factors. The factorization \eqref{eq:2} can be computed by the non-deterministic Algorithm~\ref{galg}. The details of this algorithm are explained in \cite{part11} but some comments are appropriate at this place.
\begin{itemize}
\item In all our algorithms, we denote concatenation of lists by the operator symbol ``$+$''. List concatenation is not commutative: The list $L_1 + L_2$ starts with the elements of $L_1$ and ends with the elements of~$L_2$.
\item By genericity of $M$, the norm polynomial $P\cj{P}$ is real and positive. Hence, it is the product of $n$ quadratic, real factors which are irreducible over~$\R$.
\item The choice of a quadratic factor in Line~5 is arbitrary. Different choices result in different factorizations. In general, there are $n!$ factorizations of the shape \eqref{eq:2}, each corresponding to a permutation of the quadratic factors of $P\cj{P}$.
\item For left polynomials with dual quaternion coefficients in our sense, right division is possible: Given two polynomials $M$, $N \in \H[t]$ with $N$ monic, there exist unique polynomials $Q$, $R \in \H[t]$ with $M = QN + R$ and $\deg R < \deg N$.
\item The dual quaternion $h_i$ in Line~6 can be computed as zero of the linear polynomial $R_i$ obtained by writing $M = QM_i + R_i$ (polynomial division). The assumptions on $M$ guarantee existence of a unique zero over the dual quaternions but the algorithm may fail at this point if these assumptions are not met.
\item We may exit the algorithm after just one iteration to find a linear right factor of $M$, that is, write the motion polynomial as $M = M'(t-h)$. This we will often do in our factorization algorithm for non-generic motion polynomials.
\end{itemize}

\begin{algorithm}
  \caption{$\mathtt{GFactor}$ }\label{galg}
  \begin{algorithmic}[1]
    \Require $M = P + \eps D \in \D\H[t]$, a monic, generic motion polynomial of degree~$n$.
    \Ensure A list $L=\left[L_1,\ldots,L_n\right]$ such that $M=L_1 \cdots L_n$.
    \Statex
    \State $L \leftarrow \left[\right]$ \Comment (empty list)
    \State $F\leftarrow   \left[M_1,\ldots,M_n\right]$ \Comment Each $M_i \in \R[t]$, $i=1,\ldots,n$ is a\\
           \hfill quadratic, irreducible factor of $P\overline{P} \in \R[t]$.
    \Repeat
    \State Choose $M_i \in F$ and set $F \leftarrow  F - \left[M_i\right]$.
    \State Compute $h_i$ such that $M_i(h_i) = M(h_i) = 0$.
    \State $L \leftarrow \left[t-h_i\right] + L$ \Comment (add $t-h_i$ to start of list)
    \State $M \leftarrow M/(t - h_i)$ \Comment (polynomial division)
    \Until{$\deg M = 0$}.
    \State \Return $L=\left[L_1,L_2,\ldots,L_n\right]$.
  \end{algorithmic}
\end{algorithm}

For later reference, we state the result of \cite[Theorem~3]{part11} as a lemma. We do this in a form that highlights the dependence of the factorization on an ordering of the norm polynomial's quadratic factors.

\begin{lem}
  \label{lem:factorization}
  Given a generic, monic motion polynomial $M$ of degree $n$ with $M\cj{M} = M_1 \cdots M_n$ and monic, quadratic and irreducible real polynomials $M_1, \ldots, M_n$, there exist rotation quaternions $h_1,\ldots,h_n$ such that $M = (t-h_1) \cdots (t-h_n)$ and $M_i = (t-h_i)(t-\cj{h_i})$ for $i = 1,\ldots,n$. Different labeling of the quadratic factors of $M\cj{M}$ give different factorizations.
\end{lem}

Here are examples of non-generic motion polynomials with exceptional factorizations.

\begin{example}
  \label{ex:no-factorization}
  The motion polynomial $M:=t^2+1+\eps\qi$ is not generic. A straightforward computation shows that no linear motion polynomials $t-h_1$ and $t-h_2$ in $\D\H[t]$ with $M=(t-h_1)(t-h_2)$ exist. The motion parameterized by $M$ is a translation with constant direction.
\end{example}

\begin{example}
  \label{ex:infinitely-many-factorizations}
  Non-generic motion polynomials with infinitely many factorizations exist. One example is $M:=t^2+1-\eps t(\qi t - \qj)$. It can be factorized as $M = (t-h_1)(t-h_2)$ where
  \begin{equation*}
    h_1 = \qk - \eps(a\qi + (b - 1)\qj),\quad
    h_2 = -\qk + \eps(a\qi + b\qj)
  \end{equation*}
  and $a$, $b$ are arbitrary real numbers. The motion parameterized by $M$ is a circular translation. Any of the infinitely many factorizations of $M$ corresponds two one leg of a parallelogram linkage that can generate this motion.

\end{example}

\begin{example}
  \label{ex:prismatic-joints}
  The motion polynomial
  $M:=t^2 - (1+\qj)t + \qj - \eps ((\qi+\qk)t-2\qk)$ can be factored
  as
  \begin{equation*}
    M =(t-1-\eps\qi)(t-\qj-\eps\qk) = (t-\qj-\eps(\qi+2\qk))(t-1+\eps\qk).
  \end{equation*}
  The polynomial factors $t-1-\eps\qi$ and $t-1+\eps\qk$ parameterize, however, translations, not rotations. The reason for this is the possibility to factor the primal part of $M$ as $t^2-(1+\qj)t + \qj = (t-1)(t-\qj)$. For $t = 1$, the motion parameterization becomes singular and the trajectories pass through infinite points.
\end{example}


We will present a method to factor even the motion polynomials of these examples into products of linear rotation polynomials. This will be made possible by allowing alterations of the given motion polynomial that change its kinematic and algebraic properties in an ``admissible'' way. This alterations are:
\begin{enumerate}
\item Multiplication of $M$ with a strictly positive real polynomial $Q$ and factorization of $QM$ instead of $M$. This is an admissible change because $M$ and $QM$ parameterize the same motion. This ``multiplication trick'' has already been used in \cite{planarf} for the factorization of planar motion polynomials.
\item Substitution of a rational expression $R/Q$ with $R$, $Q \in \R[t]$ for the indeterminate $t$ in $M$ and factorization of $Q^{\deg M} M(R/Q)$ instead of $M$. This amounts to a not necessarily invertible re-parameterization of the motion. In particular, it is possible to parameterize only one part of the original motion.
\end{enumerate}

Multiplication with real polynomials does not change kinematic properties but gives additional flexibility to find factorizations in otherwise unfactorizable cases. In order to explain the meaning and necessity of substitution of real polynomials, we first give an important definition.

\begin{mydef}
  A motion polynomial $M = P + \eps D$ is called \emph{bounded,} if its primal part $P$ has no real zeros.
\end{mydef}

Generic motion polynomials are bounded. Bounded motion polynomials parameterize precisely the rational motions with only bounded trajectories. If the motion polynomial is not bounded, zeros of the primal part belong to infinite points on the trajectories. For this reason, unbounded motion polynomials can never be written as the product of linear rotation polynomials.  For example, we can never succeed in finding a factorization $(t-h_1)(t-h_2)$ with rotation quaternions $h_1$, $h_2$ of the motion polynomial in Example~\ref{ex:prismatic-joints} as it has unbounded trajectories.

Unbounded motion polynomials can always be turned into bounded ones by an appropriate substitution. This is the reason, why we henceforth restrict our attention to bounded motion polynomials. The kinematic meaning is that only a certain portion of the original trajectories is actually reached during the motion. Finally, we assume that our motion polynomials are monic. This is no loss of generality. If $M$ is bounded, the leading coefficient $c_n$ of $M$ is invertible and we may factor $c_n^{-1}M$ instead. This amounts to an admissible change of coordinates.

To summarize and give a precise problem statement: Given a bounded, monic motion polynomial $M$, we want to find a real polynomial $Q$ and a list of linear rotation polynomials $L = \left[t-h_1,\ldots,t-h_n\right]$ such that $QM = (t-h_1)\cdots(t-h_n)$. In this case we say that ``$M$ admits a factorization''. We will not only prove existence of $Q$ and $L$, we will also provide a simple algorithm for computing appropriate $Q$ and $L$, provide a bound on the degree of $Q$ (and hence also on the number of polynomials in $L$) and present a more elaborate algorithm that produces a polynomial $Q$ of minimal degree.

\subsection{Factorization of non-generic cases}

On particular case for which existence of factorizations of non-generic motion polynomials has already been proved to exist is planar kinematics \cite{planarf}.

\begin{mydef}
  A motion polynomial $M$ is called \emph{planar,} if it parameterizes a planar motion (a subgroup consisting of all rotations around axes parallel to a fixed direction and translations orthogonal to that direction).
\end{mydef}

Examples of planar motion polynomials are obtained by picking coefficients in $\langle 1, \qi, \eps\qj, \eps\qk\rangle$. In \cite{planarf}, the authors showed that for every monic, bounded, planar motion polynomial $M$ of degree $n$ a real polynomial $Q$ of degree $\deg Q \le n$ exists such that $QM$ admits a factorization of the shape \eqref{eq:2}. Input and output of this planar factorization algorithm are displayed in Algorithm~\ref{p0alg}. We list this algorithm only for the purpose of later reference. For details we refer to \cite{planarf}.

\begin{algorithm}
  \caption{\texttt{PFactor} (planar factorization algorithm of \cite{planarf})}
  \label{p0alg}
  \begin{algorithmic}[1]
    \Require $M = P + \eps D \in \D\H[t]$, a planar, bounded, monic motion polynomial.
    \Ensure $Q \in \R[t]$, list $L = \left[L_1,L_2,\ldots,L_n\right]$ of linear rotation polynomials such that $QM = L_1L_2 \cdots L_n$.
  \end{algorithmic}
\end{algorithm}

The first factorization procedure we propose is of theoretical interest. It is displayed in Algorithm~\ref{talg}. It is based on the algorithm for factorization of planar motion polynomials and produces a real polynomial $Q$ and a factorization of $QM$ for a monic and bounded but not necessarily generic motion polynomial $M$. It is conceptually simpler than Algorithm~\ref{dtalg} below but non optimal as far as minimality of $\deg Q$ is concerned. In its listing, we denote by $\GRPF(M)$ the greatest real polynomial factor of a quaternion polynomial $M \in \H[t]$. Lines~\ref{talg:ifdegT} to \ref{talg:endifdegT} of Algorithm~\ref{talg} are based on the factorization
\begin{equation*}
  MT\cj{T} = (R_1T + \eps D)T\cj{T} = (R_1T\cj{T} + \eps D\cj{T})T
\end{equation*}
of $MT\cj{T}$ into the product of a planar motion polynomial and a
polynomial $T \in \H[t]$.

\begin{algorithm}
  \caption{$\mathtt{FactorI}$}\label{talg}
  \begin{algorithmic}[1]
    \Require $M=P+\eps D \in \D\H[t]$, a monic, bounded motion polynomial with real quadratic factor in its primal part, $Q \in \R[t]$, list $L$ of linear motion polynomials. Initially, $Q = 1$ and $L = \left[\right]$ (empty list).
    \Ensure $Q$ and $L=\left[L_1,L_2,\ldots,L_n\right]$ such that $QM=L_1L_2 \cdots L_n$.
    \Statex
    \State Write $P=R_1T$ where $R_1 = \GRPF(P)$.
    \If{$\deg T \neq 0$} \label{talg:ifdegT} 
        \State $L \leftarrow L + \mathtt{GFactor}(T)$ \Comment Append linear factors of $T$ to $L$.
        \State $Q \leftarrow    T\overline{T}$, 
               $P \leftarrow    R_1T\overline{T}$, 
               $D \leftarrow    D\overline{T}$, and
               $M \leftarrow P+\eps D$ \label{talg:assign}
    \EndIf \label{talg:endifdegT}
    \State Factor $MP = (P+\eps(D_1 \qi+D_2 \qj+D_3 \qk))P
                      = (P+\eps D_1\qi)(P+\eps D_2 \qj+\eps D_3 \qk)$.
    \State $Q_1, L_1 = \mathtt{PFactor}(P+\eps D_1\qi)$ \Comment (planar factorization)
    \State $Q_2, L_2 = \mathtt{PFactor}(P + \eps D_2\qj + \eps D_3\qk)$ \Comment (planar factorization)
    \State $Q \leftarrow QQ_1Q_2 = QP^2$ \Comment (because $Q_1 = Q_2 = P$)
    \State $L \leftarrow L_2 + L_1 + L$ \Comment Concatenate lists of linear factors.
    \State \Return $Q$, $L$
  \end{algorithmic}
\end{algorithm}

Together with \cite{planarf}, Algorithm~\ref{talg} proves existence of a factorization:

\begin{thm}\label{theorem:factorization}
  Given a bounded, monic motion polynomial $M \in \D\H[t]$ there always exists a real polynomial $Q$ such that $QM$ can be written as a product of linear rotation polynomials.
\end{thm}

\subsection{Factorizations of minimal degree}

Now we should further elaborate on the minimal possible degree of the real factor 
 $Q$ that makes factorization possible. In the planar case, Algorithm~\ref{p0alg} 
 gives the bound $\deg Q \le \deg M$ and this bound is known to be optimal \cite{planarf}. 
 The upper bound achievable with Algorithm~\ref{talg} is worse. 
 Let $m = \deg M$ and $r = \deg R_1$.  
 Then, the degree of $Q$ in Line~\ref{talg:assign} is bounded by 
 $2(m-r)$ and the degree of $P$ in Line~4 is bounded by $r+2(m-r)=2m-r$. 
 Hence, the degree of $Q$ at the end of Algorithm~\ref{talg} is bounded by 
 $2(m-r)+2(2m-r) = 6m-4r$. Because of $r \ge 2$, this gives the bound $\deg Q \le 6m-8$. 
 However, also in the spatial case the bound $\deg Q \le \deg M$ holds true. 
 This is guaranteed by Algorithm~\ref{dtalg}.

\begin{algorithm}
  \caption{$\mathtt{FactorAll}$}\label{dtalg}
  \begin{algorithmic}[1]
    \Require $M=P+\eps D \in \D\H[t]$, a monic, bounded motion
    polynomial of complexity $(\al,\be,\ga)$, $Q \in \R[t]$, lists $L_l$, $L_r$ of linear motion
    polynomials. Initially, $Q = 1$, $L_l = \left[\right]$, $L_r = \left[\right]$.
    \Ensure $Q$, $L_l$, $L_r$ such that with $L_l + L_r =
    \left[L_1,L_2,\ldots,L_n\right]$ we have $QM=L_1L_2\cdots L_n$.
    \Statex
    \If{$P$ has no real factors}
        \State \Return $Q$, $L_l$, $L_r + \mathtt{GFactor}(M)$. \label{dtalg:gfactor}
    \EndIf
    \State Let $R_1$ be the $\GRPF$ of $P$, i.e., $P=R_1T$, $\deg P = \beta$.
    \State Let $\al := \deg(\gcd(P,\cj{P},D\cj{D}))=\deg(\gcd(R_1,D\cj{D}))$. \Comment $\comp(M)=(\al,\be,\ga)$.
    \If{$\gcd(R_1,D\cj{D})=1$ ($\al=0$)}
        \If{$\gcd(R_1,T\cj{T})=1$}
            \If{$T = 1$, i.e., $P$ is real}
                \State Let $P_1$ be a quadratic real divisor of $P$, i.e., $P = P_1P'$. \label{dtalg:qrdP1Pp1}
                \State Compute quaternion roots $h_r$, $h_l$ of $P_1$ such that \label{dtalg:hrhl1}
                \State $h_l \neq \cj{h_r}$, $D(t-\cj{h_r}) = (t-h_l)D'$, \label{dtalg:hrhl}\Comment (Lemma~\ref{lem:factorization}, Lemma~\ref{fliplem})
                \State $(t-h_l)D'(t-h_r) = DP_1$. \label{dtalg:hrhl2}
                \State $Q \leftarrow QP_1$, $L_l \leftarrow L_l + \left[ t - h_l \right]$, 
                     $L_r \leftarrow \left[ t - h_r \right] + L_r$, \label{dtalg:hrhl2.5}
                \State $M' \leftarrow P'(t-\cj{h_l})(t-\cj{h_r})+\eps D'$.  \Comment $\comp(M')=(0,\be-2,\ga)$.
                \State \Return $\mathtt{FactorAll}(M', Q, L_l, L_r)$ \label{dtalg:qrdP1Pp2}
            \Else
                \State Let $P_1$ be a quadratic real divisor of $T\cj{T}$. \label{dtalg:qrdTT1}
                \State Compute a common zero $h$ of $P_1$ and $M$ such that
                \State $P_1 = (t-\cj{h})(t-h)$, $M = M'(t-h)$. \Comment $\comp(M')=(0,\be,\ga-1)$.
                \State \Return $\mathtt{FactorAll}(M', Q, L_l, \left[ t - h \right] + L_r)$ \label{dtalg:qrdTT2}
            \EndIf
        \Else
            \State Let $P_1$ be a quadratic real divisor of $\gcd(R_1, T\cj{T})$, i.e., $P = P'P_1$. \label{dtalg:qrdR1TT1}
            \State Compute quaternions roots $h_r$, $h_l$ of $P_1$ such that \label{dtalg:hrhl3}
            \State $P_1(h_r) = 0$, $T(h_r) \neq 0$, $T(h_l) \neq 0$, \Comment (Lemma~\ref{lem:factorization}, Lemma~\ref{fliplem})
            \State $DP_1 = D(t-\cj{h_r})(t-h_r) = (t-h_l)D'(t-h_r)$.  \label{dtalg:hrhl4} \Comment (Lemma~\ref{lem:factorization}, Lemma~\ref{fliplem})
            \State $Q \leftarrow QP_1$, $L_l \leftarrow L_l + \left[ t - h_l \right]$, 
                   $L_r \leftarrow \left[ t - h_r \right] + L_r$, \label{dtalg:hrhl4.5}
            \State $M' \leftarrow (t-\cj{h_l})P'(t-\cj{h_r}) + \eps D'$.  \Comment $\comp(M')=(0,\be-2,\ga)$.
            \State \Return $\mathtt{FactorAll}(M', Q, L_l, L_r)$ \label{dtalg:qrdR1TT2}
        \EndIf
    \Else \quad($\al\geq 2$)
        \State Let $P_1$ be a quadratic real divisor of $\gcd(R_1, D\cj{D})$. \label{dtalg:qrdR1DD}
        \State Compute quaternion roots $h_r$, $h_l$ of $P_1$ such that \label{dtalg:hrhl5} \Comment (Lemma~\ref{lem:factorization})
        \State $D = (t - h_l) D_l = D_r (t - h_r)$ and
         $P = (t - h_l) P_l = P_r (t - h_r)$ \label{dtalg:hrhl6} \Comment (Lemma~\ref{lem:factorization})
        \If{$\deg\GRPF(P_l) \le \deg\GRPF(P_r)$}
        \State $L_l \leftarrow L_l + \left[ t - h_l \right],$
         $M' \leftarrow P_l + \eps D_l$.  \Comment $\comp(M')=(\al-2,\be-2,\ga-1)$.
        \State \Return $\mathtt{FactorAll}(M', Q, L_l, L_r)$
        \Else
        \State $L_r \leftarrow \left[ t - h_r \right] + L_r$
        $M' \leftarrow P_r + \eps D_r$.  \Comment $\comp(M')=(\al-2,\be,\ga-1)$\\\hfill or $\comp(M')=(\al-2,\be-2,\ga-1)$.
        \State \Return $\mathtt{FactorAll}(M', Q, L_l, L_r)$ \label{dtalg:qrdR1DD2}
        \EndIf
    \EndIf
  \end{algorithmic}	 	
\end{algorithm}

Here are a few remarks on Algorithm~\ref{dtalg}.
  
\begin{itemize}

\item In Algorithm~\ref{dtalg}, we mainly treat the case where the primal part $P$ of the motion polynomial $M = P + \eps D$ has a non-constant real factor $R_1 = \GRPF(P)$. Otherwise, we just resort to factorization of generic motion polynomials (Algorithm~\ref{galg}).

\item The \emph{complexity} of 
a monic bounded motion polynomial 
 $M = P + \eps D \in \D\H[t]$ in Algorithm~\ref{dtalg} is a triple of integers 
 \begin{align*}
   \comp(M)  &:= (\al,\be,\ga),\\
          \al &:= \deg(\gcd(P,\cj{P},D\cj{D})),\\
          \be &:= \deg(\gcd(P,\cj{P})),\\
          \ga &:= \deg(P),
  \end{align*}
  where $\deg(a)$ is the degree of the polynomial $a$ and $\gcd(a,b) \in \R[t]$ is the greatest real common factor of polynomials $a$ and $b$. With this definition, $\gcd(a, \cj{a})$ is the greatest real polynomial factor of $a$. In each step of the recursive Algorithm~\ref{dtalg}, we try to construct $M'$ such that $\comp(M')<\comp(M)$ with lexicographic order, e.g., $(4,2,5)<(4,4,3), (4,2,2)<(4,2,3)$. Then we recursively call 
\texttt{FactorAll} with $M'$ as argument. As soon as $\al=\be=\ga=0$, Algorithm~\ref{dtalg} terminates.
 
\item The computation of quaternions $h_l$ and $h_r$ in Lines~\ref{dtalg:hrhl5}--\ref{dtalg:hrhl6} is based on Lemma~\ref{lem:factorization} and \cite[Theorem~3.2]{huang02}. One of this theorem's statements is that the set of quaternion roots of the irreducible quadratic polynomial $Q = t^2 + bt + c \in \R[t]$ is
  \begin{equation}
    \label{eq:3}
    \Bigl\{ \frac{1}{2}\bigr(-b + \sqrt{4c-b^2}(x_1\qi + x_2\qj + x_3\qk)\bigl) \mid (x_1,x_1,x_3) \in S^2 \Bigr\}
  \end{equation}
  where $S^2$ is the unit $2$-sphere in $\R^3$. In particular, for every unit vector $(x_1,x_2,x_3) \in S^2$, there is a the quaternion root $q$ whose vector part is proportional to $x_1\qi + x_2\qj + x_3\qk$. Also note that $Q = (t-h)(t-\cj{h})$ if $h$ is a quaternion root of~$Q$. In the algorithm, we can pick an arbitrary zero $h_r$ of $P_1$ and compute $D_r$ by polynomial division. Then we compute $\cj{h_l}$ as zero of the remainder polynomial $\tilde{R}$ in the division $\cj{D} = \tilde{Q}\tilde{M} + \tilde{R}$ with $\tilde{M} = (t-h_r)(t-\cj{h_r})$, as in one iteration of Algorithm~\ref{galg} and $\cj{D_l}$ again by polynomial division.

\item The computation of quaternions $h_l$ and $h_r$ in Lines~\ref{dtalg:hrhl1}--\ref{dtalg:hrhl2} 
 and Lines~\ref{dtalg:hrhl3}--\ref{dtalg:hrhl4} of Algorithm~\ref{dtalg} is again based 
 on Lemma~\ref{lem:factorization} but also on Lemma~\ref{fliplem} below. 
 Consider, for example, the situation in Lines~\ref{dtalg:hrhl1}--\ref{dtalg:hrhl2}. 
 We may prescribe $h_r$ arbitrarily as a root of $P_1$, see \eqref{eq:3}. 
 Then we use polynomial division (over $\D\H$) to find $\tilde{Q},\tilde{R} \in \H[t]$ 
 such that $(t-h_r)\cj{D} = \tilde{Q}P_1 + \tilde{R}$ and compute $\cj{h_l}$ 
 as unique zero of the linear remainder polynomial $\tilde{R}$. Using polynomial 
 division once more, we then find $\cj{D'}$ such that $(t-h_r)\cj{D} = \cj{D'}(t - \cj{h_l})$.

\end{itemize}

\begin{lem}\label{fliplem}
  Let $Q \in \R[t]$ be a quadratic polynomial that is irreducible over 
  $\R$, $D \in \H[t]$ a polynomial with $\gcd(D\cj{D}, Q) = 1$ and $O$ the 
  set of quaternion roots of $Q$. Then the map $f_{Q,D}\colon O \to O$, $h_l \mapsto h_r$ 
  with $h_r$ being the common root of $(t-h_l)D$ and $Q$ is a well-defined bijection. 
  Moreover, $f_{Q,D}(h) \neq h$ for all $h \in O$.
\end{lem}

\begin{proof}
  Our proof is based on results of \cite{part11} that state that the quaternion roots of a polynomial $P \in \H[t]$ are also roots of the quadratic factors of $P\cj{P}$. Moreover, $h$ is a root of $P$ if and only if $t-h$ is a right factor of $P$ \cite[Lemma~2]{part11}.

  By \eqref{eq:3}, the set $O$ is not empty.  The norm polynomial of $(t-h_l)D$ has the quadratic factor $Q$. Hence, there exists a quaternion root $h_r \in O$ of $(t-h_l)D$. This root is unique because of $\gcd(D\cj{D}, Q) = 1$ and the map $f_{Q,D}$ is well-defined.

  If $f_{Q,D}(h) = h$ for some $h \in O$, there exists $D' \in \H[t]$ with $D = (t-\cj{h})D'(t-h)$ and we get a contradiction to $\gcd(D\cj{D},Q) = 1$:
  \begin{equation*}
    D\cj{D} = (t-\cj{h})D'(t-h)(t-\cj{h})\cj{D'}(t-h) = Q(t-\cj{h})D'\cj{D'}(t-h).
  \end{equation*}

  By a linear parameter transformation $t \mapsto at+b$ with $a,b \in \R$ we can always achieve that $Q$ is a real multiple of $t^2 + 1$. Hence, it is no loss of generality to assume $Q = t^2 + 1$ when proving bijectivity of $f_{Q,D}$. Using polynomial division we find $K \in \H[t]$ and $a,b\in\R$ with $D= K(t^2+1)+at+b$. Then we have
  \begin{align*}
    (t-h_l)D&=(t-h_l)K(t^2+1)+(t-h_l)(at+b)\\
            &=((t-h_l)K+a)(t^2+1)+(b-h_la)t-a-h_1b.
  \end{align*}
  As already argued, there is $h_r = f_D(h_l) \in O$ such that
  \begin{equation}\label{eq:4}
    (b-h_la)h_r-a-h_lb=0.
  \end{equation}
  If there is $h'_l \neq h_l$ with $f_D(h'_l) = h_r$ then we also have
  \begin{equation}\label{eq:5}
    (b-h'_la)h_r-a-h'_lb=0.
  \end{equation}
  Subtracting Equations~\ref{eq:4} and \ref{eq:5} yields
  \begin{equation}\label{eq:6}
    (h'_l - h_l)ah_r + (h'_l - h_l)b = 0.
  \end{equation}
  As $h'_l - h_l \neq 0$, we have $ah_r + b = 0$ and this implies $D = K(t^2+1)+a(t-h_r)$. But then $\deg\gcd(D\cj{D}, t^2+1) > 0$ would contradict our assumptions. Hence $f_D$ is injective. To prove surjectivity, observe that for any $h_r \in O$, there is $h_l$ such that $(t-h_r)\overline{D}=\overline{D'}(t-h_l)$ by injectivity of $f_{\cj{D}}$. But then we have $f_D(h_l) = h_r$.
\end{proof}

The termination of Algorithm~\ref{dtalg} is guaranteed by the following theorem.
\begin{thm}\label{theorem:termination}
  Algorithm~\ref{dtalg} terminates.
\end{thm}

\begin{proof}
  The termination of the Algorithm~\ref{dtalg} is based on the reduction of the complexity 
  $\comp(M)$.  As one can see from the comments in the Algorithm~\ref{dtalg}, 
  after each recursive step $\comp(M')$ of the new motion polynomial 
  $M'$ strictly decreases. Furthermore, Lines~\ref{dtalg:qrdTT1}--\ref{dtalg:qrdTT2} 
  can not happen continually because of $\be\leq\ga$ in each motion polynomial. Then in finitely many steps we can reduce $\al$ and $\be$ to zero. After this the algorithm will terminate in one step using Algorithm~\ref{galg}.
\end{proof}

\subsection{A comprehensive example}
\label{sec:example}

Now we illustrate Algorithm~\ref{dtalg} by a comprehensive example where we really enter each 
 sub-branch once. We wish to factor the motion polynomial $M = P + \eps D$ where
\begin{equation}
  \label{eq:7}
  \begin{aligned}
    P &= (t^2+2t+2)(t^2+1)^2,\\
    D &= -(t^2+2t+2)\qi+(t^5+t^4+2t^3+t^2-t-1)\qj+(t^4+t^2-2t-1)\qk.
  \end{aligned}
\end{equation}

\itstep{First iteration} The input to Algorithm~\ref{dtalg} is $M^{(1)} = P^{(1)} + \eps Q^{(1)}$ where $P^{(1)} = P$ and $D^{(1)} = D$ from \eqref{eq:7}. We compute
\begin{equation*}
  R_1 = \GRPF(P^{(1)}) = P^{(1)}, \quad
  T  = 1, \quad
  \comp(M^{(1)}) = (2, 6, 6).
\end{equation*}
Thus, we have to use the branch in Lines~\ref{dtalg:qrdR1DD}--\ref{dtalg:qrdR1DD2} of Algorithm~\ref{dtalg}:
\begin{alignat}{2}
  h_l &= -1-\qi,&\quad h_r &= -1+\qi,\notag\\
  P_l &= (t^2+1)^2(t-\qi+1),&\quad P_r &= (t^2+1)^2(t+\qi+1),\label{eq:8}\\
  D_l &= \qj t^4 + 2 \qj t^2 - (\qi + \qj + \qk) t - 1 - \qi - \qj,&\quad D_r &= \qj t^4 + 2 \qj t^2 - (\qi + \qj + \qk) t + 1 - \qi - \qj.\label{eq:9}
\end{alignat}
Note on computation:
\begin{itemize}
\item We compute one quaternion root $h_l$ of $R_1$ by \eqref{eq:3}. We then have $R_1 = (t - h_l)(t - \cj{h_l})$ and use polynomial division to find $Q$ and $R$ with $D = Q(t-h_l)(t-\cj{h_l}) + R$.  The dual quaternion $h_r$ is the zero of the linear remainder polynomial $R$.
\item The polynomials $P_l$ and $P_r$ are also computed by polynomial division from
  \begin{equation*}
    P^{(1)} = P_r(t-h_r)\quad\text{and}\quad
    \cj{P^{(1)}} = \cj{P_l}(t-\cj{h_l}).
  \end{equation*}
  A similar computation yields $D_l$ and~$D_r$.
\end{itemize}
The updated values of $Q$, $L_l$ and $L_r$ are $Q = 1$, $L_l = [l_1] $, $L_r = [\;]$ where $l_1 = t+1+\qi$.

\itstep{Second iteration} The input to Algorithm~\ref{dtalg} is $M^{(2)} = P^{(2)} + \eps Q^{(2)}$ where $P^{(2)} = P_l$, $D^{(2)} = D_l$ are taken from \eqref{eq:8} and \eqref{eq:9}. We compute
\begin{equation*}
  R_1 = \GRPF(P^{(2)}) = (t^2+1)^2, \quad
  T  = t - \qi + 1, \quad
  \comp(M^{(2)}) = (0, 4, 5).
\end{equation*}
Because of $\gcd(R_1, D^{(2)}\cj{D^{(2)}}) = \gcd(R1, T\cj{T}) = 1$
and $T \neq 1$, we have to use the branch in Lines~\ref{dtalg:qrdTT1}--\ref{dtalg:qrdTT2} of
Algorithm~\ref{dtalg}. Using \eqref{eq:3} and polynomial division, we find
\begin{equation}
  \label{eq:10}
  \begin{gathered}
  P_1 = t^2 + 2t + 2,\quad
  h = -1 + \qi - \tfrac{39}{25}\eps\qj - \tfrac{2}{25}\eps\qk,\\
  M' = t^4
     - \tfrac{2}{25}\eps(7\qj+\qk)t^3
     + (2+\tfrac{12}{25}\qj\eps+\tfrac{16}{25}\eps\qk)t^2
     - \tfrac{8}{25}\eps(3\qj+4\qk)t
     + 1 - \eps(\qi+\tfrac{33}{25}\qj-\tfrac{31}{25}\qk).
  \end{gathered}
\end{equation}
The updated values of $Q$, $L_l$, and $L_r$ are $Q = 1$, $L_l = [l_1]$, $L_r = [t - h]$ where $r_3 = t + h$ and $h$ is as in \eqref{eq:10}.

\itstep{Third iteration} The input to Algorithm~\ref{dtalg} is $M^{(3)} = P^{(3)} + \eps Q^{(3)}$ where $M^{(3)} = M'$ is taken from \eqref{eq:10}. We compute
\begin{equation*}
  R_1 = \GRPF(P^{(3)}) = (t^2+1)^2, \quad
  T  = 1, \quad
  \comp(M^{(3)}) = (0, 4, 4).
\end{equation*}
Because of $\gcd(R_1, D^{(2)}\cj{D^{(2)}}) = 1$ and $T = 1$, we have to use the branch in Lines~\ref{dtalg:qrdP1Pp1}--\ref{dtalg:qrdP1Pp2} of Algorithm~\ref{dtalg}. Similar to the first iteration we compute
\begin{equation}
  \label{eq:11}
  \begin{gathered}
    P_1 = t^2+1,\quad P' = t^2+1,\quad h_l = \tfrac{3}{7}\qi + \tfrac{6}{7}\qj - \tfrac{2}{7}\qk,\quad
    h_r = -\qi,\\
    D' = (-\tfrac{14}{25}\qj-\tfrac{2}{25}\qk)t^3
    + (\tfrac{16}{35}-\tfrac{8}{35}\qi+\tfrac{104}{175}\qj-\tfrac{4}{25}\qk)t^2\\
    - (\tfrac{16}{35}-\tfrac{8}{35}\qi+\tfrac{188}{175}\qj+\tfrac{12}{25}\qk)t + \tfrac{24}{35} - \tfrac{67}{35}\qi-\tfrac{51}{175}\qj-\tfrac{43}{175}\qk.
  \end{gathered}
\end{equation}
The updated values of $Q$, $L_l$, and $L_r$ are $Q = t^2+1$, $L_l = [l_1,l_2]$, $L_r = [r_2, r_3]$ where
\begin{equation*}
  l_2 = t-\tfrac{3}{7}\qi-\tfrac{6}{7}\qj+\tfrac{2}{7}\qk,\quad
  r_2 = t+\qi.
\end{equation*}

\itstep{Fourth iteration} The input to Algorithm~\ref{dtalg} is $M^{(4)} = P^{(4)} + \eps D^{(4)}$ where $P^{(4)} = P'(t-\cj{h_l})(t-\cj{h_r})$ and $D^{(4)} = D'$ are taken from \eqref{eq:11}. We compute
\begin{equation*}
  R_1 = \GRPF(P^{(4)}) = t^2+1, \quad
  T = t^2 - (\tfrac{4}{7}\qi-\tfrac{6}{7}\qj+\tfrac{2}{7}\qk)t+\tfrac{3}{7}+\tfrac{2}{7}\qj+\tfrac{6}{7}\qk, \quad
  \comp(M^{(4)}) = (0, 2, 4).
\end{equation*}
Because of $\gcd(R_1, D^{(4)}\cj{D^{(4)}}) = 1$ and $\gcd(R_1, T\cj{T}) = t^2+1$, we have to use the branch in Lines~\ref{dtalg:qrdR1TT1}--\ref{dtalg:qrdR1TT2} of Algorithm~\ref{dtalg}. Similar to the first iteration we compute
\begin{equation}
  \label{eq:12}
  \begin{gathered}
    P_1 = t^2+1,\quad
    P' = t^2 - (\tfrac{4}{7}\qi - \tfrac{6}{7}\qj + \tfrac{2}{7}\qk)t + \tfrac{3}{7} + \tfrac{2}{7}\qj + \tfrac{6}{7}\qk,\\
    h_l = -\tfrac{158}{483}\qi - \tfrac{218}{483}\qj - \tfrac{401}{483}\qk\quad
    h_r = -\qk,\\
    D' = (-\tfrac{14}{25}\qj-\tfrac{2}{25}\qk)t^3
       + (\tfrac{4}{69}-\tfrac{56}{575}\qi+\tfrac{196}{345}\qj+\tfrac{8}{345}\qk)t^2\\
       - (\tfrac{28}{75}-\tfrac{44}{575}\qi+\tfrac{428}{345}\qj+\tfrac{2096}{1725}\qk)t
       - \tfrac{2308}{1725}-\tfrac{2069}{1725}\qi-\tfrac{613}{1725}\qj+\tfrac{553}{575}\qk.
  \end{gathered}
\end{equation}
The updated values of $Q$, $L_l$, and $L_r$ are $Q = (t^2+1)^2$, $L_l = [l_1,l_2,l_3]$, $L_r = [r_3, r_2, r_1]$ where
\begin{equation*}
  l_3 = t + \tfrac{158}{483}\qi + \tfrac{218}{483}\qj + \tfrac{401}{483}\qk,\quad
  r_1 = t+\qk.
\end{equation*}

\itstep{Fifth iteration} The input to Algorithm~\ref{dtalg} is $M^{(5)} = P^{(5)} + \eps D^{(5)}$ where $P^{(5)} = (t - \cj{h_l})P'(t - \cj{h_r})$ and $D^{(5)} = D'$ are taken from \eqref{eq:12}. Because of $R_1 = 1$, we have to use Line~\ref{dtalg:gfactor} of Algorithm~\ref{dtalg} and can compute a factorization of $M^{(5)}$ by means of Algorithm~\ref{galg}. Because of $M^{(5)}\cj{M^{(5)}} = (t^2+1)^4$, the factorization is unique. We find $M^{(5)} = f_1f_2f_3f_4$ where
\begin{gather*}
  f_1 = t - \tfrac{158}{483} \qi - \tfrac{218}{483} \qj - \tfrac{401}{483} \qk - \tfrac{29}{280} \eps\qi - \tfrac{37}{56} \qj\eps + \tfrac{2}{5} \eps\qk, \\
  f_2 = t + \tfrac{3}{7} \qi + \tfrac{6}{7} \qj - \tfrac{2}{7} \qk + \tfrac{43}{35} \eps\qi - \tfrac{48}{175} \qj\eps + \tfrac{51}{50} \eps\qk, \\
  f_3 = t - \qi - \tfrac{3}{2}\eps\qk,\quad
  f_4 = t - \qk - \tfrac{9}{8}\eps\qi + \tfrac{3}{8}\eps\qj.
\end{gather*}

Algorithm~\ref{dtalg} terminates and the polynomial $QM$ is the
product of the ten linear factors $l_1$, $l_2$, $l_3$, $f_1$, $f_2$, $f_3$, $f_4$, $r_1$, $r_2$, $r_3$.

\subsection{Degree bound of $Q$}
\label{sec:bound}
An upper bound on the degree of $Q$ as returned
by Algorithm~\ref{dtalg} can be read from the following theorem. 
 This degree bound is already know to be optimal.
 It is attained by certain planar motions \cite{planarf}.

\begin{thm}\label{theorem:degree}
  The degree of $Q$ as returned by Algorithm~\ref{dtalg} is less or equal to 
  the degree of the $\GRPF$ of the primal part of~$M$.
\end{thm}

\begin{proof}
  The proof follows from a careful inspection of Algorithm~\ref{dtalg}.  
  The increase of the degree of $Q$ happen either in Lines~\ref{dtalg:hrhl2}--\ref{dtalg:hrhl2.5} 
  or lines \ref{dtalg:hrhl4}--\ref{dtalg:hrhl4.5}.  Furthermore, 
  the increase of the degree of $Q$ and the decrease of the degree of 
  the $\GRPF$ are equal at these places.
\end{proof}

We illustrate Theorem~\ref{theorem:degree} by one further example. 
 One achieves the upper bound of Theorem~\ref{theorem:degree}, the other does not.

\begin{example}
  The first example is the general Darboux motion considered in \cite{SL}. Let $M = \xi P - \qi \eta \eps P \in \D\H[t]$ with
  \begin{equation*}
    \xi = t^2 + 1,\quad
    \eta = \frac{5}{2} t - \frac{3}{4},\quad
    P = t-h\quad\text{and}\quad
    h=\frac{7}{9} \qi - \frac{4}{9} \qj + \frac{4}{9} \qk.
  \end{equation*}
  As seen in \cite{SL}, this give us the factorization $M = Q_1 Q_2 Q_3$, where
  \begin{equation*}
    \begin{aligned}
      Q_1 & = t - \frac{7}{9} \qi - \frac{4}{9} \qj + \frac{4}{9} \qk - \frac{5}{4} \eps \qi + \frac{43}{64} \eps \qj - \frac{97}{64} \eps \qk, \\
      Q_2 & = t + \frac{7}{9} \qi + \frac{4}{9} \qj - \frac{4}{9} \qk,                                                                          \\
      Q_3 & = t - \frac{7}{9} \qi + \frac{4}{9} \qj - \frac{4}{9} \qk
      - \frac{5}{4} \eps \qi - \frac{43}{64} \eps \qj + \frac{97}{64}
      \eps \qk.
    \end{aligned}
  \end{equation*}
  Here, no multiplication with a real polynomial is necessary.
\end{example}

\begin{example}
  The second example is the vertical Darboux motion which was avoided in \cite{SL}. Let $M = \xi P - \qi \eta \eps P \in \D\H[t]$ with
  \begin{equation*}
    \xi = t^2 + 1,\quad
    \eta = \frac{5}{2} t - \frac{3}{4},\quad
    P = t-\qi.
  \end{equation*}
  As seen in \cite{SL}, no factorization of the shape $M = Q_1 Q_2 Q_3$ with linear motion polynomials $Q_1$, $Q_2$, $Q_3$ exists. However, we can find a factorization by multiplying with a real polynomial whose degree equals the degree of $\xi$, the greatest real polynomial factor of the primal part of~$M$. We have $(t^2 + 1)M = Q_7 Q_6^2 Q_5 Q_4,$ where
  \begin{equation*}
    Q_7 = t-\qj-\frac{3}{4}\eps \qk, \quad
    Q_6 = t+\qj-\frac{5}{4}\eps\qi+\frac{3}{8}\eps\qk,\quad
    Q_5 = t-\qj,\quad
    Q_4 = P=t-\qi.
  \end{equation*}
\end{example}

\subsection{Factorizations in planar motion groups}

Algorithm~\ref{dtalg} can produce non-planar factorizations for planar motion polynomials. 
 This is an interesting feature but may not always be desirable. If one wishes to find a 
 factorization $(t-h_1)\cdots,(t-h_n)$ of a motion polynomial in a planar motion group, say 
 $\langle 1, \qi, \eps\qj, \eps\qk \rangle$, with rotation quaternions $h_1,\ldots,h_n$ in 
 that group, we have to pick suitable left and right factors $h_l$ and $h_r$ in Algorithm~\ref{dtalg}.

 

Note that for a planar motion in the subgroup $\langle 1,\qi,\eps\qj,\eps\qk\rangle$, 
 the primal part and the dual part of a motion have a certain
 commutativity property. If $P$ is a polynomial with coefficients in
 $\langle 1,\qi \rangle$ and $D$ is a polynomial with coefficients in
 $\langle \eps\qj, \eps\qk \rangle$, then $PD = D\cj{P}$, e.g., 
 $(t-\qi) \eps\qj=\eps\qj(t+\qi)$ or $(t-\qi) \eps\qk=\eps\qk(t+\qi)$.  This allows to 
 transform right factors into left factors and vice versa. Moreover, from Equation~\ref{eq:3} it 
 follows that there are exactly two roots of a real irreducible quadratic polynomial $Q$ 
 in the planar motion subgroup. We have, for example, 
 $Q = t^2 +1 = (t-\qi)(t+\qi) = (t+\qi)(t-\qi)$. Thus, whenever we compute a quaternion 
 root of a quadratic irreducible polynomial in Algorithm~\ref{dtalg}, we should select 
 a solution in the planar motion group and whenever we transfer a left factor $h_l$ 
 to a right factor $h_r$ we should do it in such a way that $h_r = \cj{h_l}$. This 
 ensures that Algorithm~\ref{dtalg} really returns a planar factorization.


\section{Acknowledgements} 
 The research was supported by the Austrian Science Fund (FWF):
W1214-N15, project DK9 and P\;26607.

\bibliographystyle{plain}
\bibliography{Factorization}

\end{document}